\documentclass[aps,twocolumn,notitlepage,pra,superscriptaddress,footinbib]{revtex4-1}

\usepackage[usenames,dvipsnames]{color}

\usepackage{gensymb} 
\usepackage{physics} 
\usepackage{amsmath} 
\usepackage{amssymb} 
\usepackage{bm} 
\usepackage{bbm} 
\usepackage{braket} 
\usepackage[version=4]{mhchem} 
\usepackage{latexsym} 
\usepackage{tensor}
\usepackage{tcolorbox}
\usepackage{mathtools}
\usepackage{mathrsfs}

\usepackage{amsthm,thmtools}
\declaretheorem[
]{theorem}
\declaretheorem[
]{lemma}

\usepackage{CJKutf8}

\newcommand{\bh}[1]{\hat{\pmb{#1}}} 
\newcommand{\bs}[1]{\pmb{#1}} 

\usepackage{graphicx} 
\graphicspath{{./figures/}} 
\usepackage{stackengine} 
\usepackage[caption=false]{subfig} 
\usepackage[inline]{enumitem} 
\setlist[enumerate,1]{label={(\roman*)}} 
\setlist{nolistsep} 
\usepackage{comment} 

\usepackage{multirow}
\usepackage{hhline}

\usepackage{hyperref}
\usepackage{natbib}
\hypersetup{
 breaklinks=true,
 colorlinks=true,
 linkcolor=blue,
 filecolor=magenta,
 urlcolor=cyan,
}
\usepackage[capitalize]{cleveref}
\crefalias{section}{appendix}

\begin{document}
\begin{CJK}{UTF8}{gbsn}
\renewcommand{\labelenumii}{\theenumii}
\renewcommand{\theenumii}{\theenumi.\arabic{enumii}.}

\newcommand{\QuICS}{Joint Center for Quantum Information and Computer Science, National Institute of Standards and Technology and
 University of Maryland, College Park, Maryland 20742, USA}
\newcommand{\JQI}{Joint Quantum Institute, National Institute of Standards and Technology and
 University of Maryland, College Park, Maryland 20742, USA}
\newcommand{\USNA}{The Volgenau Department of Physics, United States Naval Academy, Annapolis, MD 21402, USA}
\newcommand{\harvard}{Department of Physics, Harvard University, Cambridge, MA 02138, USA}

\newcommand{\thetitle}{Multiparameter function estimation for general Hamiltonians}

\title{\thetitle}

\author{Erfan~Abbasgholinejad}
\affiliation{\QuICS}
\affiliation{\JQI}
\author{Sean~R.~Muleady}
\affiliation{\QuICS}
\affiliation{\JQI}
\author{Jacob~Bringewatt}
\affiliation{\USNA}
\affiliation{\harvard}
\author{Lorcán~O.~Conlon}
\affiliation{\QuICS}
\affiliation{\JQI}
\author{Alexey~V.~Gorshkov}
\affiliation{\QuICS}
\affiliation{\JQI}

\date{\today}

\begin{abstract}

Estimation of physical parameters encoded in a Hamiltonian is a central task in quantum sensing and learning. While the ultimate precision limit for estimating a single parameter coupled to a single generator is well established, the corresponding bound for estimating a function of multiple parameters—each coupled to distinct and possibly non-commuting generators—remains unknown in general. Here, we derive the ultimate quantum limit and present an estimation protocol for any function of parameters in a general Hamiltonian that attains this bound. We show that, although the task is fundamentally a multiparameter problem, our tight bound reduces to an optimized single-parameter quantum Cramér–Rao bound, even for arbitrary generator sets. Our result unifies and extends previous works, providing a general framework for optimal function estimation in quantum systems.
\end{abstract}

 \maketitle
\end{CJK}

\section{Introduction}

Fundamental limits constrain how precisely physical parameters can be inferred from quantum evolution~\cite{giovannetti2011advances,degen_quantum_2017}.
A central result in quantum metrology is that, for a single unknown, time-independent parameter, the achievable sensitivity is set by the corresponding signal generator via the (Helstrom) quantum Cramér–Rao bound, with theoretically optimal protocols saturating this limit~\cite{bollinger_optimal_1996,boixo2007generalized, braunstein1994statistical, pang2014quantum}. A variety of related bounds hold for the more complicated problem of estimating multiple unknown parameters, but, in general settings, the known bounds are either not tight (e.g. multiparameter generalizations of the Helstrom Cram\'er-Rao bound~\cite{helstrom1967minimum,helstrom1968minimum,yuen1973multiple} or the Gill-Massar bound~\cite{gill2005state}) or difficult to compute in closed form (e.g. the Holevo Cram\'er-Rao bound~\cite{holevo1973statistical,holevo2011probabilistic} or the Nagaoka-Hayashi bound~\cite{nagaoka2005new,nagaoka2005generalization,conlon2021efficient}).

Here, we consider the intermediate problem of estimating a single \emph{function} of multiple encoded parameters in the signal Hamiltonian. 
While we seek to estimate a single quantity, each parameter couples to a distinct generator, encompassing a genuine multiparameter problem. In this context, the fundamental question is therefore: \emph{what is the ultimate saturable quantum limit for estimating an arbitrary function of multiple parameters encoded in a general Hamiltonian?} With its single and multi-parameter aspects, it is unclear \emph{a priori} if general bounds for the function estimation problem should be essentially single parameter in nature or more similar to the known multiparameter bounds. Despite progress in specialized cases such as arbitrary commuting generators~\cite{abbasgholinejad2025optimally} or networked sensing architectures~\cite{proctor2017networked,proctor_multiparameter_2018,ehrenberg_minimum_2022,bringewatt_optimal_2024,eldredge_optimal_2018,suzuki2020quantum,Zhang2021:DQSrvw, hamann2022approximate, bate_experimental_2025, bothwell_resolving_2022, komar_quantum_2014, rubio2020quantum, bringewatt_protocols_2021}, the attainable limit for general generators is not known.

In this work, we resolve this problem fully. Generalizing the method in Ref.~\cite{eldredge_optimal_2018}, we show that any function of parameters for any encoding Hamiltonian can be mapped to an \emph{effective single-parameter} estimation problem. We derive the corresponding ultimate precision bound by optimizing over all admissible single-parameter embeddings. We prove that the bound is tight and construct an explicit protocol—based on controlled reshaping of the Hamiltonian and GHZ-like superpositions—that saturates it for any function of interest.

Our results provide a complete and unified characterization of function estimation in quantum systems, extending metrological optimality to arbitrary generator structures and enabling a complete understanding of the single and multi-parameter aspects of the problem. Indeed, while the encoding process depends genuinely on multiple parameters and may involve noncommuting generators, the ultimate saturable limit takes the form of a single-parameter quantum Cramér–Rao bound. Nonetheless, like multiparameter limits such as the Holevo Cram\'er-Rao bound and the Nagaoka-Hayashi bound, the resulting bound is not explicit and requires solving a (possibly non-trivial) optimization problem in the form of a semidefinite program \cite{holevo1973statistical,holevo2011probabilistic, nagaoka2005new,nagaoka2005generalization,conlon2021efficient, albarelli2019evaluating}.

\section{Problem Statement}
Consider a quantum system evolving under a Hamiltonian, dependent on a vector of unknown, time-independent parameters $\bs \theta \in \mathbb{R}^r$:
\begin{equation}
    \hat{H}(s)
    = \hat{H}_0 + \hat{H}_c(s),
    \qquad
    \hat{H}_0 = \sum_{j=1}^{m} f_j(\bs\theta)\, \hat{g}_j ,
    \label{eq:H}
\end{equation}
where $r\leq m$, and $\{f_j(\bs\theta)\}_{j=1}^m$ are twice continuously differentiable real functions. We assume $\{\hat{g}_j\}_{j=1}^m$ are linearly independent, traceless Hermitian generators acting on an $N$-dimensional Hilbert space. The term $\hat{H}_c(s)$ denotes a $\bs \theta$-independent, controllable, time-dependent Hamiltonian acting on the system and any ancillary system.

Rather than estimating the full parameter vector $\bs{\theta}$,  
our goal is to determine the optimal precision to which a single known real function $q(\bs \theta)$ can be estimated, given access to arbitrary state preparation, arbitrarily fast control operations generated
by $\hat{H}_c(s)$, and arbitrary measurements at the end of the evolution.
Given a fixed total evolution time $t$, our figure of merit is the minimal variance of any unbiased estimator of $q$. 

To start, we will restrict our attention to linear functions. In particular, let $f_j(\bs{\theta})=\theta_j$ with $r=m$ and consider the estimation of a linear function $q=\bs{\alpha}\cdot\bs{\theta}$ for a given vector $\bs{\alpha}\in\mathbb{R}^m$. After deriving optimal precision bounds and explicit, saturable protocols for this simplified problem, we demonstrate that the general problem reduces to an equivalent linear problem of the form already considered, with negligible overhead asymptotically in time.

\section{Derivation of the Lower Bound}\label{sec:qcrb}
We first derive the lower bound for the linear function estimation problem with arbitrary generators. 
In prior works that considered special cases of our general function estimation problem (and always with commuting generators)~\cite{eldredge_optimal_2018,proctor2017networked,proctor_multiparameter_2018,ehrenberg_minimum_2022,bringewatt_optimal_2024}, precision bounds were derived in terms of the tightest possible \emph{single-parameter} quantum Cramér–Rao bound for the quantity $q=\bs\alpha\cdot\bs{\theta} = \sum_{j=1}^m \alpha_j\theta_j$, under the assumption that the remaining $m-1$ unknown degrees of freedom in the problem are fixed. Because this approach assumes we possess additional information that we do not actually have (namely, the values of the fixed parameters), one obtains a lower bound on the minimum achievable sensitivity. By maximizing over all possible choices for fixing these extra degrees of freedom one obtains the tightest such bound. Saturability of the resulting bound can be established either via algebraic conditions~\cite{ehrenberg_minimum_2022} or through construction of an explicit, optimal protocol achieving this sensitivity. We adopt a similar strategy for our more general problem, taking the latter approach to establishing saturability.

To begin we rewrite the encoding Hamiltonian as
\begin{align}\label{eq:H_new_basis}
\hat{H}(s) &=\sum_{j=1}^m ({\bs{\alpha}}^{(j)}\cdot{\bs{\theta}})({\bs{\beta}}^{(j)}\cdot\hat{\bs{g}})+\hat{H}_c(s),
\end{align}
where $\{{\bs{\alpha}}^{(j)}\}_{j=1}^m$ are some choice of basis vectors for $\mathbb{R}^m$ such that ${\bs{\alpha}}^{(1)}={\bs{\alpha}}$, and $\{{\bs{\beta}}^{(j)}\}_{j=1}^m$ forms a dual basis such that ${\bs{\alpha}}^{(i)}\cdot{\bs{\beta}}^{(j)}=\delta_{ij}$. That is, we have performed a basis transformation ${\bs{\theta}}\rightarrow\bs{q}$, where $q_j:={\bs{\alpha}}^{(j)}\cdot{\bs{\theta}}$, with $q_1=q$ our function of interest.

If we assume that $q_j$ for $j>1$ are known, the generator for $q$ is simply
\begin{equation}\label{eq:generator}
\hat{g}_{q,{\bs{\beta}}}:=\frac{\partial \hat{H}}{\partial{q}}\Bigg|_{q_2,\cdots,q_m}={\bs{\beta}}\cdot\hat{\bs{g}},
\end{equation}
where $\bs{\beta} := {\bs{\beta}}^{(1)}$. Thus, the choice of fixing the extra degrees of freedom corresponds to the choice of $\bs{\beta}$.
Optimizing over $\bs{\beta}$, we obtain the tightest single parameter quantum Cramér–Rao bound for the function $q$ as
\begin{align}\label{eq:QCRB_linear}
\mathrm{Var}(q_\mathrm{est})\geq  \max_{{\bs{\beta}}} \frac{1}{\mathcal{F}(q|{\bs{\beta}})} \geq \max_{{\bs{\beta}}}\frac{1}{t^2\norm{\hat{g}_{q,{\bs{\beta}}}}_s^2}.
\end{align}
Here, $q_\mathrm{est}$ is an unbiased estimator for $q$, and $\mathcal{F}(q|{\bs{\beta}})$ is the quantum Fisher information with respect to the parameter $q$ with generator $\hat{g}_{q,\bs{\beta}}$.
The quantum Fisher information can be bounded by the state-independent quantity $\norm{\hat{g}_{q,{\bs{\beta}}}}_s^2$, where $\norm{\cdot}_s$ is the operator seminorm defined as the difference between the maximum and minimum eigenvalues \cite{boixo2007generalized}.
Thus, determining the lower bound on the attainable sensitivity reduces to the following optimization problem:
\begin{align}\label{eq:minprob_app1}
    \text{min}_{\bs{\beta}} \quad & \norm{{\bs{\beta}} \cdot \hat{\bs g}}_{s}, \nonumber\\
    \text{s.t.} \quad & {\bs{\alpha}}\cdot{\bs{\beta}}=1.
\end{align}
This is a semidefinite program (SDP) related, via its dual, to the minimization problem (details in \cref{app:qcrb})
\begin{align}\label{eq:bound_1norm}
    \gamma(\bh{g}|\bs{\alpha}) := \text{min}_{\hat{A}} \quad & \|\hat{A}\|_1/2, \nonumber\\
    \text{s.t.} \quad &
    \Tr\left(\hat{A}\hat g_j\right) = \alpha_j, \quad j=1,\dots,m, \nonumber\\
    & \Tr\left(\hat{A}\right)=0,
\end{align}
where $\hat{A}$ is an $N\times N$ Hermitian matrix, $\|\hat{A}\|_1 = \tr(\sqrt{\hat{A}^\dag\hat{A}})$ is the trace-norm, and $\gamma(\bs{g}|\bs{\alpha})$ is the solution to the SDP with generators $\bh{g}$ conditioned on $\bs{\alpha}$. 

Since $\hat A$ is Hermitian, it is diagonalizable in some orthonormal basis  
$\mathcal{K}=\{\ket{k}\}_{k=0}^{N-1}$, with $\hat A = \sum_k a_k\,\dyad{k}$.
Thus, for any fixed choice of basis $\mathcal{K}$, the optimization in
Eq.~\eqref{eq:bound_1norm} reduces to a linear program over the set of its diagonal entries, $\mathcal{A}_{\mathcal{K}} = \{\bs a \in \mathbb{R}^N| \sum_k a_k \bra{k}\hat{g}_j\ket{k} = \alpha_j \;  \forall j, \sum_k a_k=0\}$. The quantum Cramér–Rao bound may then be rewritten as
\begin{equation}
\label{eq:QCRB}
    \mathrm{Var}(q_{\mathrm{est}}) \ge \frac{\gamma^2(\bs{g}|\bs{\alpha})}{t^2}
    \;=\;
    \min_{\mathcal{K}}
    \;\min_{\bs a\in\mathcal{A}_\mathcal{K}}
        \frac{\|\bs a\|_1^{\,2}}{4t^{2}}.
\end{equation}
The outer minimization over $\mathcal{K}$ is implicitly restricted to those bases for which the constraints admit at least one feasible solution, i.e., for which $\mathcal{A}_\mathcal{K} \neq \varnothing$. The tracelessness and linear independence of the generators ensure that such bases always exist.

For a fixed basis $\mathcal{K}$, the matrices  
$\hat g_j' := \sum_k \bra{k}\hat g_j\ket{k}\,\dyad{k}$  
are diagonal and therefore mutually commuting.  
Equation~\eqref{eq:QCRB} is precisely the commuting-generator bound in Ref.~\cite{abbasgholinejad2025optimally} applied to the effective generators $\{\hat g_j'\}$ in basis $\mathcal{K}$ (see Table~\ref{tab:comparison}).  
If the original generators $\{\hat g_j\}$ commute, the minimizing basis is simply the common eigenbasis.

An important implication of this connection to the commuting bound  is that if we can project the Hamiltonian onto its diagonal in the optimal basis $\mathcal{K}$, then the problem becomes equivalent to one with commuting effective generators $\hat g_j'$ as mentioned above. In this commuting setting, the optimal protocol of~\cite{abbasgholinejad2025optimally} applies directly and saturates the corresponding bound. In the next section, we discuss how such a diagonal projection can be implemented.

\begin{table*}[t]
\centering
\begin{tabular}{|c|c|c|c|}
\hline
\textbf{Setting} 
& \textbf{Hamiltonian} 
& \textbf{Quantity Estimated} 
& \textbf{Precision Bound} \\ 
\hline\hline

\textbf{Single generator~\cite{boixo2007generalized}}
& $\displaystyle \hat{H} = \theta\,\hat{g}$ 
& $\displaystyle q=\theta$
& $\displaystyle 
    \frac{1}{t^{2}\, \|\hat{g}\|_{s}^{2}}
$ 
\\ \hline

\textbf{Single-qubit generators~\cite{eldredge_optimal_2018}}
& $\displaystyle \hat{H}=\sum_{j=1}^{m} \theta_j\, \hat{Z}_j$
& $\displaystyle q=\bs{\alpha}\!\cdot\!\bs{\theta}$
& $\displaystyle 
    \frac{\|\bs \alpha\|^2_\infty}{4t^{2}} $ 
\\ \hline

\textbf{Commuting generators~\cite{abbasgholinejad2025optimally}}
& $\displaystyle \hat{H}=\sum_{j=1}^{m} \theta_j\, \hat{g}_j, \: [\hat{g}_i, \hat{g}_j] = 0$
& $\displaystyle q=\bs{\alpha}\!\cdot\!\bs{\theta}$
& $\displaystyle 
    \min_{\bs a\in\mathcal{A}_\mathcal{K}}
        \frac{\|\bs a\|_1^{\,2}}{4t^{2}} $ (\cref{eq:QCRB})
\\ \hline

\textbf{Arbitrary generators} [this work]
& $\displaystyle \hat{H}=\sum_{j=1}^{m} \theta_j\, \hat{g}_j$
& $\displaystyle q=\bs{\alpha}\!\cdot\!\bs{\theta}$
& $\displaystyle 
    \frac{\gamma^2(\bs{g}|\bs{\alpha})}{t^2} $ (\cref{eq:QCRB})
\\ \hline

\textbf{Arbitrary functions} [this work]
& $\displaystyle \hat{H}=\sum_{j=1}^{m} f_j(\bs \theta)\, \hat{g}_j$
& $\displaystyle q(\bs \theta)$
& $\displaystyle 
    \frac{\gamma^2\big(F\bs{g}|\nabla q \big)}{t^2}$ (\cref{eq:QCRB_general})
\\ \hline

\end{tabular}
\caption{
Comparison of optimal precision bounds for estimating functions of
Hamiltonian parameters.  
The single-parameter case depends on the seminorm 
$\|\hat{g}\|_s$.  
Our work provides the saturable optimal bound for arbitrary, possibly 
noncommuting, generator sets, expressed through an SDP.}
\label{tab:comparison}
\end{table*}

\section{Projecting a Hamiltonian onto its Diagonal in an Arbitrary Basis}\label{sec:Hamiltonian_reshaping}
Here, we introduce a method to project a Hamiltonian onto its diagonal elements in an arbitrary orthonormal basis $\mathcal{K} = \{\ket{k}\}_{k=0}^{N-1}$. As discussed earlier, the goal is to transform the generators of the Hamiltonian into commuting ones. Thus, by the bound in Eq.~\ref{eq:QCRB}, one can employ the optimal protocol introduced in Ref.~\cite{abbasgholinejad2025optimally}. 

In the basis $\mathcal{K}$, we can write the original Hamiltonian as
\begin{equation}
    \hat{H}_0 = \sum_{k,k'} h_{k, k'} \dyad{k}{k'},
\end{equation}
where $h_{k, k'} = \bra{k} \hat{H}_0 \ket{k'}$. The goal is to construct a finite family of unitaries $\mathcal{S} =\{\hat{U}_s\}_{s=0}^{N-1}$, such that the averaged conjugation of $\hat{H}_0$ by the elements of $\mathcal{S}$ removes all off-diagonal terms of $\hat{H}$, yielding its diagonal part:
\begin{equation}
    \hat{H}_{\mathrm{eff}}=\frac{1}{|\mathcal{S}|}\sum_{s\in S} \hat{U}_{s}^{\dagger} \hat{H}_0 \hat{U}_{s}
    = \sum_k h_{k,k}\,|k\rangle\langle k|.
    \label{eq:twirl}
\end{equation}

One such choice is provided by diagonal unitaries whose phases form the character set of the cyclic group $\mathbb{Z}_N$, i.e.
\begin{equation}
    \hat{U}_{s} = \sum_{k=0}^{N-1} \omega^{s k}\,|k\rangle\langle k|,
    \quad \omega = e^{2\pi i/N}, \quad s=0,\dots,N-1.
    \label{eq:discrete_unitaries}
\end{equation}
To see this, note that conjugation of $\hat{H}_0$ by $\hat{U}_{s}$ introduces a relative phase, $\hat{U}_{s}^{\dagger} \hat{H}_0 \hat{U}_{s}
    = \sum_{k,k'} h_{k,k'}\, \omega^{s(k'-k)}\,|k\rangle\langle k'|$. By averaging this expression over $s$, we can exploit the orthogonality of the characters, $\sum_{s=0}^{N-1}\omega^{s(k'-k)} = N\delta_{k,k'}$, to cancel all off-diagonal terms and obtain Eq.~\eqref{eq:twirl}.
Thus, the discrete family of unitaries in Eq.~(\ref{eq:discrete_unitaries}) implements an exact \emph{dephasing channel} in the $\{|k\rangle\}$ basis, projecting any operator onto its diagonal components.

With access to time evolution under $\hat{H}_0$, we can approximate the evolution under $\hat{H}_{\rm eff}$ via a randomized Trotterization protocol, such as QDRIFT~\cite{ma_learning_2024, chen_concentration_2021, campbell_random_2019, childs2019faster}. As in Ref.~\cite{abbasgholinejad2025optimally}, we divide our evolution time into $L = t/\Delta t$ Trotter steps of duration $\Delta t$. During the $l$-th step, we evolve under the conjugated Hamiltonian $\hat{U}^\dagger_{s_l} \hat{H}_0\hat{U}_{s_l}$, where $\hat{U}_{s_l}$ is randomly sampled uniformly from $\mathcal{S}$. This conjugated evolution can be generated using fast pulses from the control Hamiltonian to apply $\hat{U}_{s_l}$ at the start of this interval, evolving with $\hat{H}_0$ for $\Delta t$, and applying $\hat{U}_{s_l}^\dagger$ at the end. The resulting evolution is described by
the randomized product formula 
\begin{align}
\hat{V} &= \prod_{l=1}^L \hat{U}_{s_l}^{\dagger}e^{-i  \hat{H}_0  \Delta t}\hat{U}_{s_l}
= \prod_{l=1}^L e^{-i \hat{U}_{s_l}^{\dagger} \hat{H}_0 \hat{U}_{s_l} \Delta t} .\label{eq:QDRIFT}
\end{align}
With high probability, this approximates the target unitary $\hat{U} = e^{-i\hat{H}_{\rm eff}t}$ to additive error $\epsilon$ when the number of steps satisfies
\begin{equation}\label{eq:trotter_steps} L = \Omega({\log(N)\lambda^2t^2}/{\epsilon^2}), 
\end{equation}
where $\lambda=\|\hat{H}_0\|$ for spectral norm $\|\cdot\|$; that is $\|\hat{U} - \hat{V}\|\leq \epsilon$~\cite{chen_concentration_2021}. Notably, despite $N$ possible elements in $\mathcal{S}$, the randomized product formula requires only logarithmic dependence on $N$ to achieve the approximate evolution.

We emphasize that \cref{eq:trotter_steps} applies to an arbitrary Hamiltonian; for an $n$-qubit system, it may, therefore, involve an exponential number of generators in $n$. However, in many practical scenarios, the Hamiltonian is generated by only a polynomial number of terms in \(n\), and, thus, one expects that a smaller, structured set of unitaries suffices to achieve the same task, as discussed in Ref.~\cite{huang_learning_2023, abbasgholinejad2025optimally}. Moreover, in the regime where $|\mathcal{S}| = \mathrm{poly}(n)$, one may employ deterministic constructions such as $p$-th order Trotter formulas~\cite{childs2021theory}. These schemes can yield improved scaling with the evolution time $t$, at the expense of a larger prefactor that typically grows with the system size.

\section{Optimal Protocol for Saturating the quantum Cramér–Rao bound}
As discussed in Sec.~\ref{sec:qcrb}, the reshaping procedure in the previous section effectively maps the original generator set $\{\hat{g}_j\}$ to a set of commuting generators $\{\hat{g}'_j\}$. This reduction enables us to apply the optimal commuting-generator protocol of Ref.~\cite{abbasgholinejad2025optimally}, which saturates the bound in Eq.~\eqref{eq:QCRB}. Note that our goal here is to establish attainability in principle; whether comparable performance can be realized under stricter experimental constraints remains an open question.

The protocol is constructed from the optimal dual variable $\hat{A}$ solving problem~\eqref{eq:bound_1norm}. Since $\hat{A}$ is Hermitian, it can be diagonalized in an orthonormal basis $\mathcal{K}=\{\ket{k}\}_{k=1}^{N}$ as
$\hat A = \sum_k a_k\,\dyad{k}$.
Next, we separate the eigenvectors according to the sign of their eigenvalues, defining
$\mathcal{X}=\{k\,|\,a_k>0\}$, and $\mathcal{Y}=\{k\,|\,a_k<0\}$.
In this notation,
\begin{equation}
    \hat{A} = \sum_{x \in \mathcal{X}} a_x \dyad{x} + \sum_{y \in \mathcal{Y}} a_y \dyad{y},
\end{equation}
so that $\hat{A}$ is decomposed into its positive and negative eigenspaces.

The protocol consists of evolving under the Hamiltonian in Eq.~\eqref{eq:twirl} while remaining within the family of states
$\ket{\psi_{x,y}(\phi)} \equiv \bigl(\ket{x}+e^{i\phi}\ket{y}\bigr) /\sqrt{2}$ for  $x\in \mathcal{X}$, and $y \in \mathcal{Y}$.
Since $\ket{x}$ and $\ket{y}$ are eigenstates of the Hamiltonian in Eq.~\eqref{eq:twirl}, the evolution changes only the relative phase $\phi$ and preserves the form of the state. The protocol proceeds by coherently switching between different states $\ket{\psi_{x,y}(\phi)}$. At appropriate times, this requires replacing $x\to x'$ and/or $y\to y'$. One unitary that implements such a change is
\begin{align}
    \hat{S}_{k,k'} &= \ket{k}\bra{k'} + \ket{k'} \bra{k} + \sum_{k''\neq k,k'} \ket{k''}\bra{k''}.
\end{align}
Because $\ket{x}$ and $\ket{y}$ are orthogonal, $\hat{S}_{x,x'}$ acts only on the left branch of $\ket{\psi_{x,y}(\phi)}$, whereas $\hat{S}_{y,y'}$ acts only on the right branch.

The switching times are chosen so as to accumulate a relative phase between the two branches of the GHZ-like state proportional to $q$. Specifically, we apply the swaps $\hat{S}_{x,x'}$ sequentially after time intervals of length $a_x t/(\sum_{x\in \mathcal{X}}a_x)$, covering all elements of $\mathcal{X}$. Simultaneously, we apply the swaps $\hat{S}_{y,y'}$ sequentially after intervals of length $a_y t/(\sum_{y\in \mathcal{Y}}a_y)$. With this choice of switching times, each branch of the superposition evolves for a total time $t$ under the effective Hamiltonian.

After the full sequence, the state is $\ket{\psi_{x'',y''}(\Phi)}$ up to an irrelevant global phase, where the relative phase is determined by the timing of the control operations:
\begin{align}
    \Phi 
    &= \sum_x \left(\frac{a_x}{\sum_{x'} a_{x'}} t\right)\, h_{x,x} - \sum_y \left(\frac{a_y}{\sum_{y'} a_{y'}} t\right)\, h_{y,y} \nonumber\\
    &= \frac{2t}{\|\hat{A}\|_1} \sum_j \theta_j \,\Tr\!\left(\hat{A}\hat g_j\right) 
     =  \frac{qt}{\gamma(\bs{g}|\bs{\alpha})},
\end{align}
where the second equality uses
$\|\hat{A}\|_1 = 2\sum_{x\in \mathcal{X}}a_x = 2|\sum_{y\in \mathcal{Y}}a_y|$ since $\Tr(\hat A)=0$.
Finally, measuring the operator $i\dyad{x''}{y''} + \text{h.c.}$ gives an estimator for $q$ with variance
\begin{equation}
    \mathrm{Var}(q_{\mathrm{est}}) 
    = \frac{\gamma^2(\bs{g}|\bs{\alpha})}{t^2},
\end{equation}
thereby saturating the bound in Eq.~\eqref{eq:QCRB}. Note that achieving this variance requires the system to evolve under the reshaped Hamiltonian. In general, however, the reshaping procedure described in Sec.~\ref{sec:Hamiltonian_reshaping} introduces an error that can bias the estimator. This bias can be made arbitrarily small by increasing $L$, so that its contribution is negligible compared with the optimal variance \cite{abbasgholinejad2025optimally}.

We emphasize that the number of control switches during the protocol is less than the rank of the optimal matrix $\hat{A}$. For any fixed basis $\mathcal{K}$, the optimization over $\bs a$ in \cref{eq:QCRB} is a linear program with $m+1$ linear constraints.  
By standard results on basic feasible solutions of linear programs \cite{bertsimas1997introduction, chvatal1983linear, boyd2004convex} 
(Also Theorem~1 in Ref.~\cite{abbasgholinejad2025optimally}),  there always exists an optimal solution with at most $m+1$ nonzero entries in $\bs a$.  
Since this holds for every choice of basis $\mathcal{K}$ with $\mathcal{A}_\mathcal{K}\neq \varnothing$, it follows that an optimal $\hat A$ can always be chosen with
\begin{equation}\label{eq:A_rank}
\mathrm{rank}(\hat A)\le m+1.
\end{equation}
This means that $|\mathcal{X}| + |\mathcal{Y}| \leq m+1$, and  the number of coherent switches during the protocol is less than $m$.

\section{Extension to arbitrary functions and couplings}\label{sec:QCRB_general}
In this section, we extend the bound we derived to the case of arbitrary parameter-dependent couplings and general functions of the parameters. Our approach follows and generalizes the method of Ref.~\cite{qian2020optimal}.
We consider a Hamiltonian of the form
\begin{equation} 
    \hat{H}_0 = \sum_{j=1}^m f_j(\bs{\theta}) \hat{g}_j,
\end{equation}
where each $f_j(\bs{\theta})$ is an arbitrary, sufficiently smooth function of the encoded parameters $\bs{\theta}\in\mathbb{R}^r$. We assume that the values of $\{f_j(\bs{\theta})\}_{j=1}^m$ uniquely determine $\bs{\theta}$, implying $r\le m$.
Our goal is to estimate a single arbitrary function $q(\bs{\theta})$ with minimum variance. As in Sec.~\ref{sec:qcrb}, we reduce the problem to an effective single-parameter estimation by fixing the remaining degrees of freedom and applying the tightest single-parameter bound. As shown in Appendix~\ref{app:qcrb}, the resulting generalized bound takes the form
\begin{equation}\label{eq:QCRB_general}
\mathrm{Var}(q_{\mathrm{est}})
\ge \frac{\gamma^2(F\bh{g}\mid \nabla q)}{t^2},
\end{equation}
where $\nabla q$ denotes the gradient of $q$ with respect to $\bs{\theta}$, and $F$ is the Jacobian matrix with matrix elements
\begin{equation}
F_{ij} = \frac{\partial f_j(\bs{\theta})}{\partial \theta_i}.
\end{equation}
At first glance, the bound in \cref{eq:QCRB_general} appears noninformative, since both $\nabla q$ and $F$ generally depend on the unknown parameters $\bs{\theta}$. However, inspired by Ref.~\cite{qian2020optimal}, this bound can be asymptotically saturated using a two-step adaptive protocol. The total interrogation time $t$ is divided into two stages, $t=t_1+t_2$. In the first stage, time $t_1$ is used to obtain a coarse estimate $\bs{\tilde{\theta}}$ of the parameters with $\|\bs \theta-\tilde{\bs \theta}\|_\infty = O(1/t_1)$ using any Hamiltonian learning algorithm with Heisenberg-limited time scaling~\cite{ma_learning_2024, huang_learning_2023, Hu2025:ansatzfreeHamLearn, brahmachari2026learning, Bakshi2024:HamLearn}. This allows to linearize both $q(\bs{\theta})$ and the couplings $f_j(\bs{\theta})$ around the estimated parameter values, reducing the problem to the linear setting of Sec.~\ref{sec:qcrb}. In particular, we can change the Hamiltonian into
\begin{equation} \label{eq:linearized_H}
    \delta \hat{H}_0 = \sum_{j} \delta\theta_j (\tilde F\hat{\bs g})_j + O\left(\frac{1}{t_1^2}\right),
\end{equation}
where $\tilde F:=F(\bs{\tilde\theta})$. 
The function of interest becomes
\begin{equation}\label{eq:linearized_q}
q(\bs{\theta}) = q(\tilde{\bs \theta}) + \nabla \tilde q \cdot \delta\bs{\theta} + O\left(\frac{1}{t_1^2}\right),
\end{equation}
where $\nabla\tilde{q}:=\nabla q(\bs{\tilde\theta})$.
Thus, applying the optimal linear protocol for time $t_2$, we can learn the function $q$ with the variance in Eq.~\eqref{eq:QCRB_general}. 

We emphasize that attaining this variance requires the linearization remainders in Eqs.~\eqref{eq:linearized_H} \eqref{eq:linearized_q} to contribute negligibly to the final estimation error. In particular, the bias and any excess variance arising from the coarse first-stage estimate must vanish relative to the leading term. As shown in Appendix~\ref{app:general_error}, in the asymptotic limit of large $t$, this is ensured by the simultaneous conditions $t_2/t_1^2\to 0$, and $t_1/t_2\to 0$. Thus, we need to choose $t_2\sim t$, and $t_1\sim t^p$ with $1/2< p < 1$.

\section{Discussion and Outlook}

To date, it has generally been difficult to formulate general, optimal protocols for the task of multiparameter estimation \cite{ragy2016compatibility,sidhu2021tight, albarelli2022probe}. In contrast, corresponding single-parameter estimation problems are relatively well understood, with numerous strategies available contending with limited control or relevant noise sources \cite{zhou2018achieving, demkowicz2017adaptive, pang2017optimal, yuan2015optimal, kaubruegger2025lieb}. Within the context of distributed sensing, the restricted task of learning a \emph{single} function of multiple parameters, as we have considered here, occupies a peculiar place. On the one hand, as this task amounts to the extraction of a single quantity from the underlying system, a natural inclination is to surmise that relevant bounds and strategies for single parameter sensing may well apply, up to some suitable weighting procedure or basis transformation. On the other hand, as repeatedly noted in prior works investigating various special cases of this problem~\cite{eldredge_optimal_2018,proctor_multiparameter_2018,gross_one_2020,bringewatt_protocols_2021,bringewatt_optimal_2024,abbasgholinejad2025optimally}, such an assumption is superficial, with the underlying dependence of our desired quantity on multiple, independent parameters affording a genuine multiparameter character to the problem.  

Our derived bounds and optimal protocol, which generally subsume these prior works, elucidate this distinction. To derive the bound in Eq.~\eqref{eq:QCRB}, we resorted to the use of the single-parameter Cram{\'e}r-Rao bound for a generator of the desired function, analogous to the derivation in Ref.~\cite{eldredge_optimal_2018} for independent generators and Ref.~\cite{abbasgholinejad2025optimally} for the case of commuting generators (see Table~\ref{tab:comparison}). In all these cases, the key distinction from a genuine single-parameter problem is that the signal generator is not necessarily unique, and depends on our assumptions about which additional degrees of freedom we assume to be fixed in our problem.
Thus, the function estimation task corresponds to a whole family of single-parameter problems, and identifying the ultimate quantum limits fundamentally relies on our ability to identify the optimal such problem. Additionally, the existence of a bound corresponding to some choice of fixing extra degrees of freedom does not a priori imply that we can, in practice, ``fix'' these quantities and, thus, achieve the associated bound. Optimal protocol design corresponds to identifying these quantities as nuisance parameters that can be effectively eliminated through a suitable control sequence~\cite{eldredge_optimal_2018,abbasgholinejad2025optimally}. 

The generalization to arbitrary, non-commuting generators injects an additional complexity into this picture: different choices for the signal generator do not necessarily commute, and thus generically posses incompatible eigenbases. Thus, our identification of the optimal single parameter problem requires the optimization of all single parameter generators in \emph{all possible bases}. In Appendix \ref{app:examples}, we explore simple examples of this optimization---the resulting optimal basis generally exhibits a highly nontrivial dependence on our function parameters and underlying generators.

Once such a basis and corresponding optimal generator are identified, there are two distinct sources of noise that potentially reduce the sensitivity of any sensing protocol. The first are nuisance parameters corresponding to off-diagonal contributions in the optimal signal basis. Our saturating protocol eliminates these through Hamiltonian reshaping techniques (QDRIFT), adopting a generalization of our approach in the case of commuting generators~\cite{abbasgholinejad2025optimally}. The second noise corresponds to nuisance parameters that are diagonal in the optimal signal basis, which we address with our prior strategy (see Ref.~\cite{abbasgholinejad2025optimally}) of careful, dynamical swapping of basis states so that the effects of these nuisance parameters are exactly canceled in the final phase accumulation.

Whether more experimentally relevant protocols can achieve the same sensitivity with substantially less control remains unknown, and it would be valuable to quantify more precisely the role and minimal required degree of control in function-estimation protocols of this kind. Relatedly, it would be interesting to understand how both the achievable protocols and the associated bounds are modified when explicit restrictions on the available control are imposed. A further practical challenge is the computational complexity of identifying the optimal basis. While this is not itself a fundamental metrological resource, it may nevertheless constitute a prohibitive bottleneck for large systems, except in settings where additional structure or symmetry can be exploited. It would therefore be especially interesting to characterize broader classes of structured generator families for which $\gamma(\hat{\bs g}\mid \bs\alpha)$ admits analytic expressions, or at least provably efficient approximations, at large system sizes. 

Finally, several natural extensions of our results remain open. One is to treat time-dependent signal Hamiltonians $\hat H_0(t,\bs\theta)$. Another is to move beyond the estimation of a single function and study the simultaneous estimation of multiple functions $q^{(a)}(\bs\theta)$, including the tradeoffs and optimal allocation of resources that arise when the corresponding optimal protocols are incompatible.

\textit{Acknowledgements.---}
We thank Anthony J. Brady, Yuxin Wang, and Ali Fahimniya for insightful discussions. E.A., L.O.C., and A.V.G.~were supported in part by ONR MURI, AFOSR MURI, NSF QLCI (award No.~OMA-2120757), DoE ASCR Quantum Testbed Pathfinder program (award No.~DE-SC0024220), NSF STAQ program, DARPA SAVaNT ADVENT, ARL (W911NF-24-2-0107), and NQVL:QSTD:Pilot:FTL. E.A., L.O.C., and A.V.G.~also acknowledge support from the U.S.~Department of Energy, Office of Science, National Quantum Information Science Research Centers, Quantum Systems Accelerator (Award No. DE-SCL0000121) and from the U.S.~Department of Energy, Office of Science, Accelerated Research in Quantum Computing, Fundamental Algorithmic Research toward Quantum Utility (FAR-Qu).
 S.R.M. is supported by the
NSF QLCI (award No. OMA-2120757). J.B. notes that the views expressed in this work are those of the author and do not reflect the official policy or position of the United States Naval Academy or any department of the United States Government.

\bibliography{ref}
\ \ 

\newpage
\appendix
\onecolumngrid

\newpage
\appendix

\section{Optimized quantum Cramér–Rao bound and SDP duality}
\label{app:qcrb}
In this Appendix we derive the optimized quantum Cram\'er--Rao bound used in Eqs.~\eqref{eq:QCRB}, \eqref{eq:QCRB_general} of the main text. Consider a controlled Hamiltonian
\begin{equation}
\hat H(s) = \hat H_0(\bs\theta) + \hat H_c(s),
\qquad
\hat H_0(\bs\theta) = \sum_{j=1}^m f_j(\bs\theta)\,\hat g_j,
\label{eq:H_app_general}
\end{equation}
where \(\bs\theta\in\mathbb{R}^r\) (\(r\le m\)) is unknown,
\(\hat H_c(s)\) is a known control Hamiltonian,
and each \(f_j(\bs\theta)\) is assumed twice-differentiable.
We aim to estimate a scalar function \(q(\bs\theta)\) with minimum variance.
The following bound holds for any locally unbiased estimator \(q_{\mathrm{est}}\).

\paragraph*{Step 1: effective single-parameter generator.}
Fix a point \(\bs\theta\) and define the Jacobian
\begin{equation}
F_{ij}(\bs\theta) := \frac{\partial f_j(\bs\theta)}{\partial \theta_i},
\qquad i=1,\dots,r,\ \ j=1,\dots,m,
\label{eq:Jacobian_def}
\end{equation}
and the gradient \((\nabla q)_i := \partial q/\partial \theta_i\).
We choose a local reparametrization \((q_1,\dots,q_r)\) with \(q_1=q(\bs\theta)\),
and consider variations in \(q_1\) while holding \(q_2,\dots,q_r\) fixed.
Along such a one-dimensional curve, the parameter variation satisfies
\begin{equation}
\frac{\partial \bs\theta}{\partial q_1} = \bs\beta,
\qquad
\nabla q \cdot \bs\beta = 1,
\label{eq:beta_constraint_app}
\end{equation}
where \(\bs\beta\in\mathbb{R}^r\) depends on the choice of fixed coordinates
\(q_2,\dots,q_r\).

By the chain rule, the generator of translations in \(q_1\) is
\begin{align}
\hat g_{q,\bs\beta}
&:= \frac{\partial \hat H_0}{\partial q_1}\Big|_{q_2,\dots,q_r}
= \sum_{i=1}^r \beta_i \frac{\partial \hat H_0}{\partial \theta_i}
= \sum_{i=1}^r \sum_{j=1}^m \beta_i\,F_{ij}(\bs\theta)\,\hat g_j.
\label{eq:gq_general_app}
\end{align}
It is convenient to define the \(r\) effective generators
\begin{equation}
\hat h_i := \frac{\partial \hat H_0}{\partial \theta_i}
= \sum_{j=1}^m F_{ij}(\bs\theta)\,\hat g_j,
\qquad i=1,\dots,r,
\label{eq:hi_def}
\end{equation}
so that \(\hat g_{q,\bs\beta} = \bs\beta\cdot \hat{\bs h}\).

\paragraph*{Step 2: optimized single-parameter bound.}
For unitary encoding generated by \(\hat g_{q,\bs\beta}\) (with arbitrary known controls),
the quantum Fisher information satisfies the standard bound~\cite{boixo2007generalized}
\begin{equation}
\mathcal F(q\mid \bs\beta)\ \le\ t^2\,\|\hat g_{q,\bs\beta}\|_s^2,
\end{equation}
implying the single-parameter quantum Cramér–Rao bound
\begin{equation}
\mathrm{Var}(q_{\mathrm{est}})
\ \ge\
\frac{1}{\,\mathcal F(q\mid \bs\beta)}
\ \ge\
\frac{1}{t^2\,\|\hat g_{q,\bs\beta}\|_s^2}.
\label{eq:sp_qcrb_app}
\end{equation}
Since the choice of the fixed coordinates \((q_2,\dots,q_r)\) is arbitrary, we may optimize
over all \(\bs\beta\) obeying \eqref{eq:beta_constraint_app}, obtaining
\begin{align}
\mathrm{Var}(q_{\mathrm{est}})
&\ge
\max_{\bs\beta:\,\nabla q\cdot\bs\beta=1}
\frac{1}{t^2\,\|\bs\beta\cdot \hat{\bs h}\|_s^2}
=
\frac{1}{t^2}\left(
\min_{\bs\beta:\,\nabla q\cdot\bs\beta=1}
\|\bs\beta\cdot \hat{\bs h}\|_s
\right)^{-2}.
\label{eq:general_bound_beta_app}
\end{align}
Defining
\begin{equation}
\gamma(\hat{\bs h}\mid \nabla q)
:=
\left(
\min_{\bs\beta:\,\nabla q\cdot\bs\beta=1}
\|\bs\beta\cdot \hat{\bs h}\|_s
\right)^{-1},
\label{eq:gamma_general_def}
\end{equation}
the optimized bound reads
\begin{equation}
{
\mathrm{Var}(q_{\mathrm{est}})
\ \ge\
\frac{\gamma^2(\hat{\bs h}\mid \nabla q)}{t^2}.
}
\label{eq:QCRB_general_app}
\end{equation}

Next, we express the seminorm minimization in \cref{eq:gamma_general_def} as an SDP.
For Hermitian \(\hat G\), the identity
\begin{equation}
\|\hat G\|_s
=
\min_{u,v\in\mathbb{R}}\ \{v-u:\ u\hat I \preceq \hat G \preceq v\hat I\}
\label{eq:seminorm_as_sdp}
\end{equation}
holds because the smallest feasible \(v\) and largest feasible \(u\) are
\(v=\lambda_{\max}(\hat G)\) and \(u=\lambda_{\min}(\hat G)\).
Substituting \(\hat G=\bs\beta\cdot \hat{\bs h}\) gives
\begin{align}
\gamma(\hat{\bs h}\mid \nabla q)^{-1}
=
\min_{\bs\beta,u,v}\quad & v-u
\nonumber\\
\text{s.t.}\quad &
u\hat I \preceq \bs\beta\cdot \hat{\bs h} \preceq v\hat I,
\nonumber\\
& \nabla q\cdot\bs\beta = 1.
\label{eq:primal_sdp_app}
\end{align}

We now derive the dual problem to \eqref{eq:primal_sdp_app}. Introduce dual variables \(\hat X\succeq 0\) and \(\hat Y\succeq 0\) for the matrix inequalities
\(\bs\beta\cdot \hat{\bs h}-v\hat I \preceq 0\) and
\(u\hat I-\bs\beta\cdot \hat{\bs h}\preceq 0\), respectively,
and a scalar \(\xi\in\mathbb{R}\) for the equality constraint \(\nabla q\cdot\bs\beta=1\).
The Lagrangian is
\begin{align}
\mathcal L(\bs\beta,u,v;\xi,\hat X,\hat Y)
&= (v-u)
+ \mathrm{Tr}\!\left[\hat X(\bs\beta\cdot \hat{\bs h}-v\hat I)\right]
+ \mathrm{Tr}\!\left[\hat Y(u\hat I-\bs\beta\cdot \hat{\bs h})\right]
+ \xi(1-\nabla q\cdot\bs\beta).
\end{align}
Minimizing over \(u\) and \(v\) forces
\begin{equation}
\mathrm{Tr}(\hat X)=1,
\qquad
\mathrm{Tr}(\hat Y)=1.
\label{eq:trace_constraints_app}
\end{equation}
Minimizing over \(\bs\beta\) yields the stationarity conditions
\begin{equation}
\mathrm{Tr}\!\left[(\hat X-\hat Y)\hat h_j\right] - \xi\,\partial_j q = 0,
\qquad j=1,\dots,r,
\end{equation}
where $\partial_j q$ is the partial derivative with respect to $\theta_j$. Thus, the dual problem is 
\begin{align}
\max_{\xi,\hat X,\hat Y}\quad & \xi
\nonumber\\
\text{s.t.}\quad &
\mathrm{Tr}\!\left[(\hat X-\hat Y)\hat h_j\right] = \xi\,\partial_j q,\quad j=1,\dots,r,
\nonumber\\
& \hat X,\hat Y \succeq 0,\qquad \mathrm{Tr}(\hat X)=\mathrm{Tr}(\hat Y)=1.
\label{eq:dual_sdp_app}
\end{align}

Next, define \(\hat Z:=\hat X-\hat Y\). From \(\mathrm{Tr}(\hat X)=\mathrm{Tr}(\hat Y)=1\),
we have \(\mathrm{Tr}(\hat Z)=0\), and by the triangle inequality,
\begin{equation}
\|\hat Z\|_1 \le \|\hat X\|_1 + \|\hat Y\|_1 = \mathrm{Tr}(\hat X)+\mathrm{Tr}(\hat Y)=2,
\label{eq:Z_norm_bound}
\end{equation}
where $\|\hat{Z}\|_1 = \tr(\sqrt{\hat{Z}^\dag\hat{Z}})$ is the trace-norm. Assuming \(\xi\neq 0\), define \(\hat A := \hat Z/\xi\). Then, the constraints in
\eqref{eq:dual_sdp_app} become
\begin{equation}
\mathrm{Tr}(\hat A\hat h_j)=\partial_j q,\quad j=1,\dots,r,
\qquad
\mathrm{Tr}(\hat A)=0,
\label{eq:A_constraints_app}
\end{equation}
and \eqref{eq:Z_norm_bound} implies
\begin{equation}
\frac12\|\hat A\|_1 = \frac12\frac{\|\hat Z\|_1}{|\xi|}\ \le\ \frac{1}{|\xi|}.
\end{equation}
Conversely, given any Hermitian \(\hat A\) obeying \eqref{eq:A_constraints_app}, let
\(\hat A=\hat A_+-\hat A_-\) be its Jordan decomposition, with \(\hat A_\pm\succeq 0\),
\(\hat A_+\hat A_-=0\), and \(\mathrm{Tr}(\hat A)=0\Rightarrow
\mathrm{Tr}(\hat A_+)=\mathrm{Tr}(\hat A_-)=:\tau\).
Then \(\tau=\|\hat A\|_1/2\). Setting
\begin{equation}
\hat X := \hat A_+/\tau,\qquad \hat Y := \hat A_-/\tau,\qquad \xi := 1/\tau
\label{eq:construct_XY_from_A}
\end{equation}
gives \(\hat X,\hat Y\succeq 0\), \(\mathrm{Tr}(\hat X)=\mathrm{Tr}(\hat Y)=1\), and
\(\mathrm{Tr}[(\hat X-\hat Y)\hat h_j]=\xi\,\partial_jq\).
Thus, maximizing \(\xi\) is equivalent to minimizing \(\tau=\|\hat A\|_1/2\), which implies the optimization problem introduced in Eq.~\eqref{eq:bound_1norm} of the main text:
\begin{align}
\gamma(\hat{\bs h}\mid \nabla q)
=
\min_{\hat{A}} \quad & \|\hat{A}\|_1/2, \nonumber\\
    \text{s.t.} \quad &
    \Tr\left(\hat{A}\hat h_j\right) = \partial_jq, \quad j=1,\dots,r, \nonumber\\
    & \Tr\left(\hat{A}\right)=0.
\label{eq:gamma_trace_norm_app}
\end{align}
The reformulation of the problem from \cref{eq:dual_sdp_app} to \cref{eq:gamma_trace_norm_app} allows us to apply the dual technique yet again in order to obtain another formulation for $\gamma(\hat{\bs h}\mid \nabla q)$. This final formulation will be useful for analyzing the effects of errors in \cref{app:error_gamma}. In particular, introducing Lagrange multipliers \(y_j \in \mathbb{R}\) for the equality constraints \(\Tr(\hat A \hat h_j)=\partial_j q\) and \(\mu \in \mathbb{R}\) for the trace constraint, the Lagrangian associated with \eqref{eq:gamma_trace_norm_app} is
\begin{equation}
\mathcal{L}(\hat A; y, \mu)
=
\frac{1}{2}\|\hat A\|_1
-\Tr\!\left[\hat A\!\left(\sum_{j=1}^r y_j \hat h_j + \mu \hat I\right)\right]
+\sum_{j=1}^r y_j \partial_j q.
\end{equation}

The dual function is obtained by minimizing the Lagrangian over \(\hat A\). Consequently, the dual optimization problem associated with
\eqref{eq:gamma_trace_norm_app} is
\begin{align}\label{eq:gamma_op_norm}
\gamma(\hat{\bs h}\mid \nabla q) = \max_{y \in \mathbb{R}^r,\; \mu \in \mathbb{R}}\quad
& \bs y \cdot \nabla q \\
\text{s.t.}\quad
& \left\| \bs y \cdot \hat{\bs h} + \mu \hat I \right\|
\le \frac{1}{2}, \nonumber
\end{align}
where $\| \cdot \|$ denotes the spectral norm (largest singular value).
Thus, we have three different optimization problems for $\gamma(\hat{\bs h}\mid \nabla q)$, which are stated in Eqs.~(\ref{eq:gamma_general_def},\ref{eq:gamma_trace_norm_app}, and \ref{eq:gamma_op_norm}). The minimization in \cref{eq:gamma_general_def} arises from the general quantum Cramér–Rao bound, and we use the optimal $\hat{A}$ in \cref{eq:gamma_trace_norm_app} to construct an optimal protocol to saturate the bound in the main text. Finally, we use the reformulation in \cref{eq:gamma_op_norm} to analyze the error, as detailed in the next section.

\section{Error analysis for the two-step adaptive protocol}
\label{app:general_error}

In this Appendix, we provide details for the error terms appearing in the two-step protocol described in
Sec.~\ref{sec:QCRB_general} (i.e. the protocol for general couplings and arbitrary analytic functions).

\subsection{Stability of \texorpdfstring{$\gamma(\hat{\bs g}\mid\bs\alpha)$}{gamma} under perturbations}
\label{app:error_gamma}

First, we quantify how the SDP solution value
\(\gamma(\hat{\bs g}\mid\bs\alpha)\) changes under small perturbations of the generators
\(\hat g_j\) and the coefficient vector \(\bs\alpha\). This result is used in Appendix~\ref{app:protocol_error} to bound the error of the two-step protocol. 
Throughout we assume \(\{\hat I,\hat g_1,\dots,\hat g_m\}\) is linearly independent over \(\mathbb{R}\),
so that the conditioning constant \(\kappa(\hat{\bs g})\) defined below is strictly positive.

\subsubsection{Setup and dual representation}

As derived in the previous section (see \cref{eq:gamma_op_norm}) the solution of the SDP of interest can be formulated as
\begin{align}\label{eq:gamma_dual_app}
\gamma(\hat{\bs g}\mid \bs \alpha) = \max_{y \in \mathbb{R}^m,\; \mu \in \mathbb{R}}\quad
& \bs y \cdot \bs \alpha \\
\text{s.t.}\quad
& \left\| \bs y \cdot \hat{\bs g} + \mu \hat I \right\|
\le \frac{1}{2}. \nonumber
 \end{align}
Consider the perturbed problem where \(\hat{\tilde{g}}_j=\hat g_j+\Delta\hat g_j\) and \(\tilde{\bs\alpha}=\bs\alpha+\Delta\bs\alpha\),for perturbations $\Delta\hat g_j$ and  $\Delta\bs\alpha$ and define the perturbed solution as 
\begin{equation}
\tilde\gamma := \gamma(\hat{\tilde{\bs g}}\mid \tilde{\bs\alpha}).
\end{equation}
We quantify the size of the perturbations by
\begin{equation}
\varepsilon_g := \max_{j}\|\Delta\hat g_j\|,
\qquad
\varepsilon_\alpha := \|\Delta\bs\alpha\|_\infty,
\label{eq:eps_defs_app}
\end{equation}
where \(\|\cdot\|\) denotes the spectral norm (maximum singular value) for operators, and for vectors we use the \(\ell_p\) norm
\begin{equation}
\|\mathbf{x}\|_p :=
\begin{cases}
\left(\sum_{i} |x_i|^p\right)^{1/p}, & 1 \le p < \infty, \\[6pt]
\max_{ i} |x_i|, & p = \infty.
\end{cases}
\end{equation}

\subsubsection{Conditioning constant and bounded dual multipliers}

Define the conditioning constant
\begin{equation}
\kappa(\hat{\bs g})
:=
\min_{\|\bs y\|_1=1}\ \min_{\mu\in\mathbb{R}}
\left\|\sum_{j=1}^m y_j \hat g_j + \mu \hat I \right\|.
\label{eq:kappa_def_app2}
\end{equation}
Via the assumption of linear independence of \(\{\hat I,\hat g_1,\dots,\hat g_m\}\),
we have \(\kappa(\hat{\bs g})>0\). We begin with a lemma that will be useful later regarding the stability of \(\kappa\) under perturbations.

\begin{lemma}[Perturbation of \(\kappa\)]
\label{lem:kappa_perturb_app}
Let \(\tilde g_j=\hat g_j+\Delta\hat g_j\) and \(\varepsilon_g\) as in \eqref{eq:eps_defs_app}. Then
\begin{equation}
\kappa(\hat{\tilde{\bs g}})\ \ge\ \kappa(\hat{\bs g}) - \varepsilon_g.
\label{eq:kappa_lipschitz_app}
\end{equation}
In particular, if \(\varepsilon_g<\kappa(\hat{\bs g})\), then \(\kappa(\hat{\tilde{\bs g}})>0\).
\end{lemma}

\begin{proof}
Fix any \(\bs y\) with \(\|\bs y\|_1=1\). For any \(\mu\),
\begin{align}
\left\|\sum_j y_j \tilde g_j + \mu \hat I\right\|
&=
\left\|\sum_j y_j \hat g_j + \mu \hat I + \sum_j y_j \Delta \hat g_j\right\|
\nonumber\\
&\ge
\left\|\sum_j y_j \hat g_j + \mu \hat I\right\|
-
\left\|\sum_j y_j \Delta \hat g_j\right\|
\nonumber\\
&\ge
\left\|\sum_j y_j \hat g_j + \mu \hat I\right\|
-
\sum_j |y_j|\,\|\Delta\hat g_j\|
\nonumber\\
&\ge
\left\|\sum_j y_j \hat g_j + \mu \hat I\right\|
-\varepsilon_g.
\end{align}
Taking \(\min_\mu\) on both sides and then \(\min_{\|\bs y\|_1=1}\) yields \eqref{eq:kappa_lipschitz_app}.
\end{proof}

The next lemma bounds the size of $\norm{\vec{y}}$ for $\vec{y}$ satisfying the constraint  in \eqref{eq:gamma_dual_app}.

\begin{lemma}[Bounded dual multipliers]
\label{lem:y_bound_app}
Let \(\kappa(\hat{\bs g})\) be defined in \eqref{eq:kappa_def_app2}. Then any \((\bs y,\mu)\) satisfying
\(\left\|\sum_j y_j \hat g_j + \mu \hat I\right\| \le \frac12\) obeys
\begin{equation}
\|\bs y\|_1 \le \frac{1}{2\,\kappa(\hat{\bs g})}.
\label{eq:y_bound_app}
\end{equation}
Similarly, any \((\bs y,\mu)\) feasible for the perturbed constraint
\(\left\|\sum_j y_j \tilde g_j + \mu \hat I\right\| \le \frac12\) obeys
\begin{equation}
\|\bs y\|_1 \le \frac{1}{2\,\kappa(\hat{\tilde{\bs g}})}
\ \le\ \frac{1}{2\,(\kappa(\hat{\bs g})-\varepsilon_g)},
\label{eq:y_bound_tilde_app}
\end{equation}
provided \(\varepsilon_g<\kappa(\hat{\bs g})\).
\end{lemma}

\begin{proof}
If \(\bs y=\bs 0\), the claim is trivial. Otherwise write \(\bs y=\|\bs y\|_1\,{\bs y'}\)
with \(\|{\bs y'}\|_1=1\). By definition of \(\kappa(\hat{\bs g})\),
\begin{equation}
\min_{\nu\in\mathbb{R}}
\left\|\sum_j  y'_j \hat g_j + \nu \hat I\right\| \ge \kappa(\hat{\bs g}).
\end{equation}
Scaling \(\nu=\mu/\|\bs y\|_1\) gives
\begin{equation}
\min_{\mu\in\mathbb{R}}
\left\|\sum_j y_j \hat g_j + \mu \hat I\right\|
=
\|\bs y\|_1
\min_{\nu\in\mathbb{R}}
\left\|\sum_j  y'_j \hat g_j + \nu \hat I\right\|
\ge
\kappa(\hat{\bs g})\,\|\bs y\|_1.
\end{equation}
Since \(\left\|\sum_j y_j \hat g_j + \mu \hat I\right\|\) is at least its minimum over \(\mu\),
feasibility implies \(\frac12\ge \kappa(\hat{\bs g})\|\bs y\|_1\), proving \eqref{eq:y_bound_app}.
The perturbed bound \eqref{eq:y_bound_tilde_app} follows identically with \(\hat{\tilde{\bs g}}\),
together with Lemma~\ref{lem:kappa_perturb_app}.
\end{proof}

\subsubsection{A scaling relation between the dual feasible sets}

Let
\begin{equation}
\mathcal D(\hat{\bs g})
:=
\left\{(\bs y,\mu):\left\|\sum_{j=1}^m y_j \hat g_j + \mu \hat I\right\| \le \frac12\right\}
\end{equation}
be the feasible set of the dual \eqref{eq:gamma_dual_app}, and define \(\mathcal D(\hat{\tilde{\bs g}})\) similarly.

\begin{lemma}[Feasibility transfer by scaling]
\label{lem:scaling_app}
Fix \((\bs y,\mu)\in\mathcal D(\hat{\bs g})\) and define
\begin{equation}
\delta(\bs y):=\left\|\sum_{j=1}^m y_j \Delta\hat g_j\right\|
\le \varepsilon_g\,\|\bs y\|_1.
\label{eq:delta_def_app}
\end{equation}
Then the scaled pair
\begin{equation}
(\bs y',\mu'):=\frac{1}{1+2\delta(\bs y)}(\bs y,\mu)
\label{eq:scaling_map_app}
\end{equation}
is feasible for the perturbed dual, i.e. \((\bs y',\mu')\in\mathcal D(\hat{\tilde{\bs g}})\).
\end{lemma}

\begin{proof}
Using \(\hat{\tilde g}_j=\hat g_j+\Delta\hat g_j\), we have
\begin{align}
\left\|\sum_j y'_j \hat{\tilde g}_j + \mu' \hat I\right\|
&=
\frac{1}{1+2\delta(\bs y)}
\left\|
\sum_j y_j \hat g_j + \mu \hat I + \sum_j y_j \Delta\hat g_j
\right\|
\nonumber\\
&\le
\frac{1}{1+2\delta(\bs y)}
\left(
\left\|\sum_j y_j \hat g_j + \mu \hat I\right\|
+
\left\|\sum_j y_j \Delta\hat g_j\right\|
\right)
\nonumber\\
&\le
\frac{1}{1+2\delta(\bs y)}
\left(\frac12+\delta(\bs y)\right)
=\frac12,
\end{align}
where we used \((\bs y,\mu)\in\mathcal D(\hat{\bs g})\) and the definition of \(\delta(\bs y)\).
\end{proof}

\subsubsection{Main perturbation bound for \texorpdfstring{$\gamma$}{gamma}}
\label{app:error_gamma:main}

The lemmas from the previous subsection in hand, we now bound the error  \(|\tilde\gamma-\gamma|\)  in the SDP solution due to perturbation in terms of \(\varepsilon_g\) and \(\varepsilon_\alpha\).

\begin{theorem}[Error scaling of \(\gamma\)]
\label{thm:gamma_perturb_app}
Let \(\kappa:=\kappa(\hat{\bs g})>0\), and assume \(\varepsilon_g<\kappa\).
Let \(\gamma=\gamma(\hat{\bs g}\mid\bs\alpha)\) and \(\tilde\gamma=\gamma(\hat{\tilde{\bs g}}\mid\tilde{\bs\alpha})\),
with \(\varepsilon_g,\varepsilon_\alpha\) as in \eqref{eq:eps_defs_app}.
Then
\begin{equation}
\frac{\gamma - \frac{\varepsilon_\alpha}{2\kappa}}{1+\frac{\varepsilon_g}{\kappa}}
\ \le\
\tilde\gamma
\ \le\
\left(1+\frac{\varepsilon_g}{\kappa-\varepsilon_g}\right)\gamma
+\frac{\varepsilon_\alpha}{2(\kappa-\varepsilon_g)}.
\label{eq:gamma_two_sided_app}
\end{equation}
Consequently,
\begin{equation}
|\tilde\gamma-\gamma|
\ \le\
\frac{\varepsilon_g}{\kappa-\varepsilon_g}\,\gamma
+\frac{\varepsilon_\alpha}{2(\kappa-\varepsilon_g)}.
\label{eq:gamma_abs_bound_app}
\end{equation}
In particular, if \(\varepsilon_g=O(\varepsilon)\) and \(\varepsilon_\alpha=O(\varepsilon)\), then
\(|\tilde\gamma-\gamma|=O(\varepsilon)\).
\end{theorem}

\begin{proof}
\emph{Lower bound.}
Let \((\bs y^\star,\mu^\star)\) be an optimizer of \eqref{eq:gamma_dual_app} for \(\gamma\), so that
\(\gamma=\bs\alpha\cdot \bs y^\star\) and \((\bs y^\star,\mu^\star)\in\mathcal D(\hat{\bs g})\).
By Lemma~\ref{lem:y_bound_app}, \(\|\bs y^\star\|_1\le 1/(2\kappa)\), hence
\(\delta(\bs y^\star)\le \varepsilon_g/(2\kappa)\).
By Lemma~\ref{lem:scaling_app}, the scaled point
\((\bs y',\mu')=(\bs y^\star,\mu^\star)/(1+2\delta(\bs y^\star))\) lies in \(\mathcal D(\hat{\tilde{\bs g}})\), so
\begin{align}
\tilde\gamma
&=
\max_{(\bs y,\mu)\in\mathcal D(\hat{\tilde{\bs g}})} (\bs\alpha+\Delta\bs\alpha)\cdot \bs y
\ \ge\
(\bs\alpha+\Delta\bs\alpha)\cdot \bs y'
\nonumber\\
&=
\frac{(\bs\alpha+\Delta\bs\alpha)\cdot \bs y^\star}{1+2\delta(\bs y^\star)}
=
\frac{\gamma + \Delta\bs\alpha\cdot \bs y^\star}{1+2\delta(\bs y^\star)}.
\end{align}
Using \(|\Delta\bs\alpha\cdot \bs y^\star|\le \|\Delta\bs\alpha\|_\infty \|\bs y^\star\|_1
\le \varepsilon_\alpha/(2\kappa)\) and \(\delta(\bs y^\star)\le \varepsilon_g/(2\kappa)\) yields
\[
\tilde\gamma \ \ge\ \frac{\gamma-\varepsilon_\alpha/(2\kappa)}{1+\varepsilon_g/\kappa},
\]
which is the left inequality in \eqref{eq:gamma_two_sided_app}.

\emph{Upper bound.}
Let \((\tilde{\bs y}^\star,\tilde\mu^\star)\) be an optimizer for \(\tilde\gamma\), so that
\(\tilde\gamma=(\bs\alpha+\Delta\bs\alpha)\cdot \tilde{\bs y}^\star\) and
\((\tilde{\bs y}^\star,\tilde\mu^\star)\in\mathcal D(\hat{\tilde{\bs g}})\).
By Lemma~\ref{lem:y_bound_app} and Lemma~\ref{lem:kappa_perturb_app},
\(\|\tilde{\bs y}^\star\|_1 \le 1/(2(\kappa-\varepsilon_g))\).
Now consider its feasibility for the \emph{unperturbed} constraint:
\begin{align}
\left\|\sum_j \tilde y^\star_j \hat g_j + \tilde\mu^\star \hat I\right\|
&\le
\left\|\sum_j \tilde y^\star_j \hat{\tilde g}_j + \tilde\mu^\star \hat I\right\|
+
\left\|\sum_j \tilde y^\star_j \Delta\hat g_j\right\|
\nonumber\\
&\le
\frac12 + \varepsilon_g\,\|\tilde{\bs y}^\star\|_1
\le
\frac12 + \frac{\varepsilon_g}{2(\kappa-\varepsilon_g)}.
\end{align}
Let \(\tilde\delta:=\left\|\sum_j \tilde y^\star_j \Delta\hat g_j\right\|
\le \varepsilon_g\,\|\tilde{\bs y}^\star\|_1\).
Scaling down by \(1+2\tilde\delta\) yields a feasible point for \(\mathcal D(\hat{\bs g})\), hence
\begin{align}
\gamma
&=
\max_{(\bs y,\mu)\in\mathcal D(\hat{\bs g})} \bs\alpha\cdot \bs y
\ \ge\
\bs\alpha\cdot \frac{\tilde{\bs y}^\star}{1+2\tilde\delta}
=
\frac{\bs\alpha\cdot \tilde{\bs y}^\star}{1+2\tilde\delta}
\nonumber\\
&=
\frac{\tilde\gamma - \Delta\bs\alpha\cdot \tilde{\bs y}^\star}{1+2\tilde\delta}.
\end{align}
Rearranging gives
\begin{equation}
\tilde\gamma
\le
(1+2\tilde\delta)\gamma + \Delta\bs\alpha\cdot \tilde{\bs y}^\star
\le
\left(1+\frac{\varepsilon_g}{\kappa-\varepsilon_g}\right)\gamma
+
\frac{\varepsilon_\alpha}{2(\kappa-\varepsilon_g)},
\end{equation}
where we used \(\tilde\delta\le \varepsilon_g\|\tilde{\bs y}^\star\|_1
\le \varepsilon_g/(2(\kappa-\varepsilon_g))\) and
\(|\Delta\bs\alpha\cdot \tilde{\bs y}^\star|\le \varepsilon_\alpha\|\tilde{\bs y}^\star\|_1
\le \varepsilon_\alpha/(2(\kappa-\varepsilon_g))\).
This proves the right inequality in \eqref{eq:gamma_two_sided_app}.

Finally, \eqref{eq:gamma_abs_bound_app} follows by combining the two-sided bounds and using that
\(\kappa-\varepsilon_g\le \kappa\) and the denominators are positive when \(\varepsilon_g<\kappa\).
\end{proof}

\noindent \paragraph*{\textbf{Remark} (simplified bound for sufficiently small perturbations).}
If \(\varepsilon_g \le \kappa/2\), then \(\kappa-\varepsilon_g\ge \kappa/2\) and \eqref{eq:gamma_abs_bound_app} implies
\begin{equation}
|\tilde\gamma-\gamma|
\le
\frac{2\gamma}{\kappa}\,\varepsilon_g + \frac{1}{\kappa}\,\varepsilon_\alpha,
\end{equation}
which makes the linear \(O(\varepsilon)\) scaling explicit.

\subsection{Error bounds for the two-step protocol}
\label{app:protocol_error}

Next, we bound the errors in the two-step protocol of Sec.\ref{sec:QCRB_general}, and show that it saturates the optimal bound in Eqs.~\eqref{eq:QCRB_general}, and \eqref{eq:QCRB_general_app}.

\subsubsection{Assumptions and notation}
Let \(\bs\theta\in\mathbb{R}^r\) be the true parameter and let \(\tilde{\bs\theta}\) be the stage-1 estimate.
Define the estimation error
\begin{equation}
\delta\bs\theta := \bs\theta-\tilde{\bs\theta},
\qquad
\|\delta\bs\theta\|_\infty = O(1/t_1).
\label{eq:delta_theta_def_app}
\end{equation}
Assume \(f_j(\bs\theta)\) and \(q(\bs\theta)\) are $C^2$ in a neighborhood
\(\mathcal U\) containing the segment between \(\bs\theta\) and \(\tilde{\bs\theta}\), and define uniform
Hessian bounds
\begin{align}
L_f &:= \max_{j}\ \max_{a,b\in\{1,\dots,r\}}\ \sup_{\bs\vartheta\in\mathcal U}
\left|\frac{\partial^2 f_j(\bs\vartheta)}{\partial \theta_a\,\partial \theta_b}\right|,
\label{eq:Lf_def_app}\\
L_q &:= \max_{a,b\in\{1,\dots,r\}}\ \sup_{\bs\vartheta\in\mathcal U}
\left|\frac{\partial^2 q(\bs\vartheta)}{\partial \theta_a\,\partial \theta_b}\right|.
\label{eq:Lq_def_app}
\end{align}
Also let
\begin{equation}
G_{\max} := \max_{j=1,\dots,m}\ \|\hat g_j\|.
\label{eq:Gmax_def_app}
\end{equation}

Recall the Jacobian \(F(\bs\theta)\in\mathbb{R}^{r\times m}\) with entries
\(F_{ij}(\bs\theta)=\partial f_j(\bs\theta)/\partial\theta_i\), and the effective generators
\begin{equation}
\hat h_i(\bs\theta) := \sum_{j=1}^m F_{ij}(\bs\theta)\,\hat g_j,
\qquad
i=1,\dots,r,
\label{eq:h_def_app}
\end{equation}
so that \(\hat{\bs h}(\bs\theta)=F(\bs\theta)\hat{\bs g}\).
Similarly, define \(\tilde F:=F(\tilde{\bs\theta})\), \(\nabla\tilde q := \nabla q(\tilde{\bs\theta})\), and
\(\hat{\tilde{\bs h}}:=\hat{\bs h}(\tilde{\bs\theta})=\tilde F\hat{\bs g}\).

Finally, assume the relevant \(\gamma\) values are finite and the conditioning constant
\(\kappa(\hat{\bs h}(\bs\theta))>0\) (\cref{eq:kappa_def_app2}), so that the stability bound for \(\gamma\)
from Appendix~\ref{app:error_gamma} applies.

\subsubsection{Taylor linearization errors}
We first quantify the remainders in Eqs.~\eqref{eq:linearized_H}--\eqref{eq:linearized_q}.

\paragraph*{Couplings.}
By multivariate Taylor's theorem, for each \(j\) there exists \(\bs\vartheta_j\) on the line segment between
\(\bs\theta\) and \(\tilde{\bs\theta}\) such that
\begin{equation}
f_j(\bs\theta)
=
f_j(\tilde{\bs\theta})
+\sum_{i=1}^r \tilde F_{ij}\,\delta\theta_i
+\frac12\,\delta\bs\theta^{\mathsf T}\,H_j(\bs\vartheta_j)\,\delta\bs\theta,
\label{eq:taylor_f_app}
\end{equation}
where \(H_j\) is the Hessian matrix of \(f_j\).
Using \eqref{eq:Lf_def_app} and \(\|\delta\bs\theta\|_1\le r\|\delta\bs\theta\|_\infty\), we obtain the bound
\begin{equation}
\left|f_j(\bs\theta)-f_j(\tilde{\bs\theta})-\sum_{i=1}^r \tilde F_{ij}\,\delta\theta_i\right|
\le \frac12\,L_f\,\|\delta\bs\theta\|_1^2
\le \frac12\,L_f\,r^2\,\|\delta\bs\theta\|_\infty^2
=
O\!\left(\frac{1}{t_1^2}\right).
\label{eq:f_remainder_bound_app}
\end{equation}

Subtracting the known term \(\sum_j f_j(\tilde{\bs\theta})\hat g_j\) by control, the unknown Hamiltonian in stage~2 can be written as
\begin{equation}
\hat H_0(\bs\theta)-\hat H_0(\tilde{\bs\theta})
=
\sum_{i=1}^r \delta\theta_i\,(\tilde F\hat{\bs g})_i
+\hat R_H,
\label{eq:H_remainder_decomp_app}
\end{equation}
where the remainder satisfies
\begin{equation}
\|\hat R_H\|
\le
\sum_{j=1}^m
\left|f_j(\bs\theta)-f_j(\tilde{\bs\theta})-\sum_{i} \tilde F_{ij}\,\delta\theta_i\right|\ \|\hat g_j\|
\le
\frac12\,m\,G_{\max}\,L_f\,r^2\,\|\delta\bs\theta\|_\infty^2
=
O\!\left(\frac{1}{t_1^2}\right).
\label{eq:R_H_bound_app}
\end{equation}
Which results in Eq.~\eqref{eq:linearized_H}.

\paragraph*{Function of interest.}
Similarly, Taylor's theorem gives
\begin{equation}
q(\bs\theta)=q(\tilde{\bs\theta})+\nabla\tilde q\cdot\delta\bs\theta + R_q,
\label{eq:taylor_q_app}
\end{equation}
with
\begin{equation}
|R_q|
\le
\frac12\,L_q\,\|\delta\bs\theta\|_1^2
\le
\frac12\,L_q\,r^2\,\|\delta\bs\theta\|_\infty^2
=
O\!\left(\frac{1}{t_1^2}\right),
\label{eq:R_q_bound_app}
\end{equation}
which justifies the \(O(1/t_1^2)\) term in Eq.~\eqref{eq:linearized_q}.

\subsubsection{Stability of {$\gamma(F\hat{\bs g}\mid \nabla q)$}}
We now relate \(\gamma(\tilde F\hat{\bs g}\mid \nabla\tilde q)\) to \(\gamma(F\hat{\bs g}\mid \nabla q)\).

\paragraph*{Perturbations in the effective linear model.}
Recall the (unknown) ideal objects at \(\bs\theta\),
\begin{equation}
\hat{\bs h}:=\hat{\bs h}(\bs\theta)=F(\bs\theta)\hat{\bs g},
\qquad
\bs\alpha:=\nabla q(\bs\theta),
\end{equation}
and their stage-1 approximations
\begin{equation}
\tilde{\hat{\bs h}}:=\hat{\bs h}(\tilde{\bs\theta})=\tilde F\hat{\bs g},
\qquad
\tilde{\bs\alpha}:=\nabla q(\tilde{\bs\theta})=\nabla\tilde q.
\end{equation}
We bound the induced perturbations
\(\Delta\hat h_i := \hat{\tilde h}_i-\hat h_i\) and \(\Delta\bs\alpha:=\tilde{\bs\alpha}-\bs\alpha\).

First, by the mean value theorem applied to \(F_{ij}=\partial f_j/\partial\theta_i\),
for each \(i,j\) there exists \(\bs\vartheta_{ij}\in\mathcal U\) such that
\begin{equation}
|\tilde F_{ij}-F_{ij}|
=
\left|\sum_{a=1}^r
\frac{\partial^2 f_j(\bs\vartheta_{ij})}{\partial\theta_i\,\partial\theta_a}\ \delta\theta_a\right|
\le
\sum_{a=1}^r L_f\,|\delta\theta_a|
\le
r L_f\,\|\delta\bs\theta\|_\infty.
\label{eq:DeltaF_bound_app}
\end{equation}
Hence, using \eqref{eq:Gmax_def_app},
\begin{align}
\|\Delta\hat h_i\|
&=
\left\|\sum_{j=1}^m (\tilde F_{ij}-F_{ij})\,\hat g_j\right\|
\le
\sum_{j=1}^m |\tilde F_{ij}-F_{ij}|\,\|\hat g_j\|
\nonumber\\
&\le
m\,G_{\max}\,r\,L_f\,\|\delta\bs\theta\|_\infty
=
O\!\left(\frac{1}{t_1}\right),
\qquad i=1,\dots,r.
\label{eq:Delta_h_bound_app}
\end{align}
Second, applying the mean value theorem to \(\nabla q\) and using \eqref{eq:Lq_def_app} gives
\begin{equation}
\|\Delta\bs\alpha\|_\infty
=
\|\nabla q(\tilde{\bs\theta})-\nabla q(\bs\theta)\|_\infty
\le
r L_q\,\|\delta\bs\theta\|_\infty
=
O\!\left(\frac{1}{t_1}\right).
\label{eq:Delta_alpha_bound_app}
\end{equation}

\paragraph*{Applying the stability theorem for \(\gamma\).}
Let
\begin{equation}
\gamma:=\gamma(\hat{\bs h}\mid \bs\alpha)=\gamma(F\hat{\bs g}\mid \nabla q),
\qquad
\tilde\gamma:=\gamma(\hat{\tilde{\bs h}}\mid \tilde{\bs\alpha})
=\gamma(\tilde F\hat{\bs g}\mid \nabla\tilde q).
\end{equation}
Define perturbation magnitudes (in the form used in Appendix~\ref{app:error_gamma})
\begin{equation}
\varepsilon_g := \max_{i=1,\dots,r}\ \|\Delta\hat h_i\|,
\qquad
\varepsilon_\alpha := \|\Delta\bs\alpha\|_\infty.
\label{eq:eps_ga_app}
\end{equation}
By \eqref{eq:Delta_h_bound_app}--\eqref{eq:Delta_alpha_bound_app},
\(\varepsilon_g=O(1/t_1)\) and \(\varepsilon_\alpha=O(1/t_1)\).

Let \(\kappa:=\kappa(\hat{\bs h})>0\) be the conditioning constant for the effective generators \(\hat{\bs h}\). For sufficiently large \(t_1\), we have \(2\varepsilon_g<\kappa\).
Then Theorem~\ref{thm:gamma_perturb_app} in Appendix~\ref{app:error_gamma} gives the bound
\begin{equation}
|\tilde\gamma-\gamma|
\ \le\
\frac{\varepsilon_g}{\kappa-\varepsilon_g}\,\gamma
+\frac{\varepsilon_\alpha}{2(\kappa-\varepsilon_g)}
=
O\!\left(\frac{1}{t_1}\right).
\label{eq:Delta_gamma_bound_app}
\end{equation}

\subsubsection{Conditions on $t_1,t_2$ and overall MSE scaling}
The stage-2 estimator targets the \emph{linearized} quantity \(\nabla\tilde q\cdot\delta\bs\theta\).
The two Taylor remainders \(\hat R_H\) and \(R_q\) induce additional errors of order \(O(1/t_1^2)\)
in the effective linear model and the function of interest.

\paragraph*{Negligibility of the Taylor remainder in \(q\).}
From \eqref{eq:R_q_bound_app}, the linear approximation error in \(q\) satisfies \(|R_q|=O(1/t_1^2)\).
If the stage-2 estimator is locally unbiased for the linear term \(\nabla\tilde q\cdot\delta\bs\theta\),
then \(R_q\) appears as a bias in the overall estimator of \(q\), contributing
\begin{equation}
\mathrm{Bias}^2(q_{\mathrm{est}})\ \le\ |R_q|^2 = O\!\left(\frac{1}{t_1^4}\right).
\end{equation}
To make this negligible compared to the stage-2 variance \(\sim 1/t_2^2\), it suffices that
\begin{equation}
\frac{1}{t_1^4} \ll \frac{1}{t_2^2}
\quad\Longleftrightarrow\quad
\frac{t_2}{t_1^2}\to 0.
\label{eq:condition_t2_t1_sq_app}
\end{equation}

\paragraph*{Negligibility of the Taylor remainder in \(\hat H_0\).}
Model the stage-2 control as known instantaneous pulses \(\{\hat V_\ell\}_{\ell=1}^{M}\) interleaved with free
evolution for durations \(\{\tau_\ell\}_{\ell=1}^{M+1}\), with \(\sum_{\ell}\tau_\ell=t_2\).
Write \(\hat H_{\mathrm{true}}=\hat H_{\mathrm{lin}}+\hat R_H\) where
\(\|\hat R_H\|=O(1/t_1^2)\).
Let \(\hat U_{\mathrm{lin}}\) and \(\hat U_{\mathrm{true}}\) be the resulting implemented unitaries:
\begin{equation}
\hat U_{\mathrm{lin}}=\hat V_M e^{-i\hat H_{\mathrm{lin}}\tau_{M+1}}\cdots \hat V_1 e^{-i\hat H_{\mathrm{lin}}\tau_{1}},
\qquad
\hat U_{\mathrm{true}}=\hat V_M e^{-i\hat H_{\mathrm{true}}\tau_{M+1}}\cdots \hat V_1 e^{-i\hat H_{\mathrm{true}}\tau_{1}}.
\end{equation}
For each segment, the Duhamel formula gives
\begin{equation}
    \|e^{-i\hat H_{\mathrm{true}}\tau}-e^{-i\hat H_{\mathrm{lin}}\tau}\| = \left\|-i\int_0^\tau
e^{-i\hat H_{\mathrm{true}}(\tau-s)}
\hat R_H
e^{-i\hat H_{\mathrm{lin}}s}
\,ds \right\| \le \tau\|\hat R_H\|,
\end{equation}
since conjugation by unitaries does not change \(\|\hat R_H\|\).
A telescoping expansion of the product therefore yields
\begin{equation}
\|\hat U_{\mathrm{true}}-\hat U_{\mathrm{lin}}\|
\le \sum_{\ell=1}^{M+1}\tau_\ell\,\|\hat R_H\|
= t_2\,\|\hat R_H\|
=O\!\left(\frac{t_2}{t_1^2}\right).
\label{eq:unitary_mismatch_pulses_short}
\end{equation}
Thus the induced change in the output state is \(O(t_2/t_1^2)\), which is
negligible under the condition \(t_2/t_1^2\to 0\).

\paragraph*{Putting the terms together.}
Under \(\varepsilon_g,\varepsilon_\alpha=O(1/t_1)\), the stability bound \eqref{eq:Delta_gamma_bound_app} implies
\(\tilde\gamma^2=\gamma^2+O(1/t_1)\).
Thus the stage-2 variance satisfies
\begin{equation}
\mathrm{Var}\!\left((\nabla\tilde q\cdot\delta\bs\theta)_{\mathrm{est}}\right)
=
\frac{\tilde\gamma^2}{t_2^2}
=
\frac{\gamma^2(F\hat{\bs g}\mid\nabla q)}{t_2^2}
+O\!\left(\frac{1}{t_1 t_2^2}\right).
\end{equation}
Adding the squared bias from the Taylor remainder \(R_q\) gives the overall mean-squared error
\begin{equation}
\mathrm{MSE}(q_{\mathrm{est}})
=
\frac{\gamma^2(F\hat{\bs g}\mid\nabla q)}{t_2^2}
+O\!\left(\frac{1}{t_1 t_2^2}\right)
+O\!\left(\frac{1}{t_1^4}\right),
\label{eq:MSE_total_app}
\end{equation}
with the additional error from \(\hat R_H\) suppressed under the condition
\(t_2/t_1^2\to 0\) discussed above.

Finally, choosing \(t_2\sim t\) and \(t_1\sim t^p\) with \(1/2<p<1\) ensures
\(t_1/t_2\to 0\) and \(t_2/t_1^2\to 0\), so that \eqref{eq:MSE_total_app} reduces asymptotically to the optimal bound \(\mathrm{MSE}(q_{\mathrm{est}})\sim \gamma^2(F\hat{\bs g}\mid\nabla q)/t^2\) in Eqs.~\eqref{eq:QCRB_general}, and \eqref{eq:QCRB_general_app}.

\section{Examples}
\label{app:examples}

In this Appendix we collect several instances in which the SDP
\begin{align}
\gamma(\hat{\bs g}\mid\bs\alpha)
:= \min_{\hat A}\quad & \frac12\|\hat A\|_1
\nonumber\\
\text{s.t.}\quad
& \Tr(\hat A\hat g_j)=\alpha_j,\quad j=1,\dots,m,
\nonumber\\
& \Tr(\hat A)=0,
\label{eq:gamma_def_examples}
\end{align}
can be solved analytically.
Throughout, we use standard Pauli operators
\(\hat X,\hat Y,\hat Z\) with eigenvalues \(\pm 1\).
For an \(n\)-qubit Pauli string \(\hat P\), we will use
\(\hat P^2=\hat I\), \(\Tr(\hat P)=0\), and \(\Tr(\hat P\hat P')=2^n\delta_{\hat P,\hat P'}\).

\subsection{Single qubit: Pauli-vector model}
\label{app:examples:1q}

Consider a single qubit with
\begin{equation}
\hat H_0=\sum_{j=1}^3 \theta_j\,\hat\sigma_j,
\qquad
q(\bs\theta)=\bs\alpha\cdot\bs\theta,
\end{equation}
where \(\hat\sigma_j\in\{\hat X,\hat Y,\hat Z\}\).
We compute \(\gamma(\hat{\bs\sigma}\mid\bs\alpha)\) from \eqref{eq:gamma_def_examples}.
Any Hermitian traceless \(\hat A\) can be written as \(\hat A=\bs a\cdot\hat{\bs\sigma}\).
Using \(\Tr(\hat\sigma_i\hat\sigma_j)=2\delta_{ij}\), the constraints \(\Tr(\hat A\hat\sigma_j)=\alpha_j\)
imply \(2a_j=\alpha_j\), hence
\begin{equation}
\hat A^\star=\frac12\,\bs\alpha\cdot\hat{\bs\sigma}.
\end{equation}
Its eigenvalues are \(\pm \|\bs\alpha\|_2/2\), so
\begin{equation}
\|\hat A^\star\|_1=\|\bs\alpha\|_2
\qquad\Longrightarrow\qquad
\gamma(\hat{\bs\sigma}\mid\bs\alpha)=\frac{\|\bs\alpha\|_2}{2}.
\end{equation}
Therefore the bound becomes
\begin{equation}
\mathrm{Var}(q_{\mathrm{est}})
\ \ge\
\frac{\|\bs\alpha\|_2^2}{4t^2}.
\end{equation}

\subsection{Complete Hilbert--Schmidt orthonormal traceless basis}
\label{app:examples:orthonormal}
Consider the Hamiltonian
\begin{equation}
    \hat H_0=\sum_{j=1}^{N^2-1} \theta_j\,\hat g_j,
    \qquad
    q(\bs\theta)=\bs\alpha\cdot\bs\theta,
\end{equation}
and let \(\{\hat g_j\}_{j=1}^{N^2-1}\) be a Hilbert--Schmidt orthonormal basis of traceless Hermitian matrices (e.g. generalized Gell-Mann matrices or properly normalized Pauli strings):
\begin{equation}
\Tr(\hat g_i\hat g_j)=\delta_{ij},
\qquad
\Tr(\hat g_j)=0.
\end{equation}
Then any Hermitian traceless \(\hat A\) has a unique expansion
\(\hat A=\sum_{j=1}^{N^2-1} a_j \hat g_j\), and the constraints
\(\Tr(\hat A\hat g_j)=\alpha_j\) enforce \(a_j=\alpha_j\).
Thus the feasible \(\hat A\) is unique:
\begin{equation}
\hat A^\star=\sum_{j=1}^{N^2-1}\alpha_j \hat g_j,
\end{equation}
and
\begin{equation}
\gamma(\hat{\bs g}\mid\bs\alpha)=\frac12\left\|\sum_{j=1}^{N^2-1}\alpha_j \hat g_j\right\|_1.
\end{equation}

\subsection{Estimating a single Pauli coefficient in a multi-Pauli Hamiltonian}
\label{app:examples:single_param_pauli}

Consider an \(n\)-qubit Hamiltonian with multiple Pauli-string generators,
\begin{equation}
\hat H_0 =\sum_{j=1}^m \theta_j\,\hat P_j,
\label{eq:H_multi_pauli}
\end{equation}
where each \(\hat P_j\) is an \(n\)-qubit Pauli string (\(\hat P_j^2=\hat I\), \(\Tr(\hat P_j)=0\)),
and assume the set \(\{\hat P_j\}\) contains no duplicates up to a sign.
We are interested in estimating only one coefficient, say \(\theta_{j_0}\), i.e.
\begin{equation}
q(\bs\theta)=\theta_{j_0},
\qquad
\bs\alpha = \bs e_{j_0},
\end{equation}
where \(\bs e_{j_0}\) is the standard basis vector.

The SDP \eqref{eq:gamma_def_examples} becomes
\begin{align}
\gamma(\{\hat P_j\}\mid \bs e_{j_0})
=\min_{\hat A}\quad & \frac12\|\hat A\|_1
\nonumber\\
\text{s.t.}\quad &
\Tr(\hat A \hat P_{j_0})=1,\qquad
\Tr(\hat A \hat P_{j})=0\ \ (j\neq j_0),
\nonumber\\
& \Tr(\hat A)=0.
\label{eq:single_coeff_sdp}
\end{align}
A feasible operator is obtained by taking
\begin{equation}
\hat A^\star=\frac{1}{2^n}\,\hat P_{j_0},
\label{eq:A_star_single_coeff}
\end{equation}
since \(\Tr(\hat P_{j_0}^2)=2^n\) gives \(\Tr(\hat A^\star \hat P_{j_0})=1\), while orthogonality of distinct Pauli
strings implies \(\Tr(\hat P_{j_0}\hat P_j)=0\) for \(j\neq j_0\), and \(\Tr(\hat A^\star)=0\).

The eigenvalues of \(\hat A^\star\) are \(\pm 2^{-n}\) each with multiplicity \(2^{n-1}\), hence
\begin{equation}
\|\hat A^\star\|_1 = 2^n \cdot 2^{-n} = 1,
\qquad\Longrightarrow\qquad
\gamma(\{\hat P_j\}\mid \bs e_{j_0}) \le \frac12.
\end{equation}
To see this is optimal, use the dual form
\[
\gamma(\hat{\bs g}\mid\bs\alpha)
=\max_{\bs y,\mu}\left\{\bs\alpha\cdot\bs y:\left\|\sum_{j=1}^m y_j \hat P_j+\mu\hat I\right\|_\infty\le \frac12\right\}.
\]
Choose \(\mu=0\) and \(\bs y=\frac12\,\bs e_{j_0}\). Since \(\|\hat P_{j_0}\|=1\), this point is feasible and yields
\(\bs\alpha\cdot\bs y = \frac12\). Therefore,
\begin{equation}
\gamma(\{\hat P_j\}\mid \bs e_{j_0}) = \frac12,
\quad\text{and}\quad
\mathrm{Var}(\theta_{j_0,\mathrm{est}})\ \geq \frac{1}{4t^2}.
\end{equation}
Notably, this value is independent of the remaining generators \(\{\hat P_j\}_{j\neq j_0}\).

\subsection{Commuting local generators: \texorpdfstring{$m$}{m} independent \texorpdfstring{$\hat Z$}{Z} operators}
\label{app:examples:mZ_commuting}

Consider the set up studied in \cite{eldredge_optimal_2018} with the Hamiltonian 
\begin{equation}
    \hat H_0 =\sum_{j=1}^m \theta_j\,\hat Z_j,
\end{equation}
where $Z_j$ is the pauli Z operator on qubit $j$, and \(q(\bs\theta)=\bs\alpha\cdot\bs\theta\).
Then
\begin{equation}
{\ \gamma(\hat Z_1,\dots,\hat Z_m\mid \bs\alpha)=\frac{\|\bs\alpha\|_\infty}{2}\ },
\qquad
\Rightarrow\qquad
\mathrm{Var}(q_{\mathrm{est}})\ge \frac{\|\bs\alpha\|_\infty^2}{4t^2}.
\end{equation}

\begin{proof}
Use the dual characterization of \(\gamma\):
\begin{equation}
\gamma(\hat{\bs g}\mid\bs\alpha)
=
\max_{\bs y\in\mathbb{R}^m,\ \mu\in\mathbb{R}}
\bs\alpha\cdot\bs y:\left\|\sum_{j=1}^m y_j\hat Z_j+\mu\hat I\right\|\le \frac12.
\label{eq:mZ_dual_short}
\end{equation}
Since all \(\hat Z_j\) commute, they are diagonal in the computational basis
\(\ket{\bs z}\) with \(\bs z\in\{\pm1\}^m\) and \(\hat Z_j\ket{\bs z}=z_j\ket{\bs z}\). Hence
\begin{equation}
\left\|\sum_{j=1}^m y_j\hat Z_j+\mu\hat I\right\|_\infty
=
\max_{\bs z\in\{\pm1\}^m}\left|\mu+\sum_{j=1}^m y_j z_j\right|.
\end{equation}
Because the spectrum contains both \(\mu+s\) and \(\mu-s\) (take \(\bs z\) and \(-\bs z\)),
the minimizer over \(\mu\) is \(\mu=0\), giving the equivalent constraint
\begin{equation} 
\max_{\bs z}\left|\sum_{j=1}^m y_j z_j\right|\le \frac12.
\end{equation}
But \(\max_{\bs z}|\sum_j y_j z_j|=\sum_j |y_j|=\|\bs y\|_1\) (choose \(z_j=\mathrm{sgn}(y_j)\)), so
\eqref{eq:mZ_dual_short} reduces to
\begin{equation}
    \gamma(\hat Z_1,\dots,\hat Z_m\mid \bs\alpha)=\max_{\|\bs y\|_1\le 1/2}\ \bs\alpha\cdot\bs y.
\end{equation}
and,
\begin{equation}
    \max_{\|\bs y\|_1\le 1/2}\ \bs\alpha\cdot\bs y=\frac12\|\bs\alpha\|_\infty,
\end{equation}
achieved by \(\bs y=\frac12\,\mathrm{sgn}(\alpha_k)\,\bs e_k\) for
\(k =\arg\max_j|\alpha_j|\). This proves the claim.
\end{proof}

\subsection{Two-qubit noncommuting example: \texorpdfstring{$\hat Z_1,\hat X_1,\hat Z_1\hat Z_2$}{Z1, X1, Z1Z2}}
\label{app:examples:2q_partial_noncommuting}

Consider the Hamiltonian 
\begin{equation}\label{eq:gens_2q_example}
    \hat{H}_0 = \theta_1 \hat Z_1 + \theta_2 \hat X_1 + \theta_3 \hat Z_1\hat Z_2,
\end{equation}
with $q(\bs\theta)=\bs\alpha\cdot\bs\theta$. Although the generators do not commute (e.g. \([\hat Z_1,\hat X_1]\neq 0\) and \([\hat X_1,\hat Z_1\hat Z_2]\neq 0\)),
they all commute with \(\hat Z_2\), which enables block diagonalization, and we can show 
\begin{equation}
\gamma(\hat Z_1,\hat X_1,\hat Z_1\hat Z_2\mid \alpha_1,\alpha_2,\alpha_3)
=
\frac12\sqrt{\alpha_2^2+\max\{|\alpha_1|,|\alpha_3|\}^2}.
\label{eq:gamma_closed_2q}
\end{equation}

\begin{proof}
We use the dual representation:
\begin{align}
\gamma(\hat{\bs g}\mid\bs\alpha)
=
\max_{\bs y\in\mathbb{R}^3,\ \mu\in\mathbb{R}}\quad &
\bs\alpha\cdot\bs y
\nonumber\\
\text{s.t.}\quad &
\left\|y_1\hat Z_1+y_2\hat X_1+y_3\hat Z_1\hat Z_2+\mu\hat I\right\|
\le \frac12 .
\label{eq:dual_2q_start}
\end{align}

\emph{Step 1: block diagonalization.}
All terms in \eqref{eq:dual_2q_start} commute with \(\hat Z_2\), so in the \(\hat Z_2\) eigenbasis the operator is
block diagonal with two \(2\times 2\) blocks acting on qubit~1:
\begin{equation}
y_1\hat Z_1+y_2\hat X_1+y_3\hat Z_1\hat Z_2+\mu\hat I
=
\begin{pmatrix}
\mu\hat I + (y_1+y_3)\hat Z + y_2\hat X & 0\\
0 & \mu\hat I + (y_1-y_3)\hat Z + y_2\hat X
\end{pmatrix}.
\end{equation}
Define
\begin{equation}
r_\pm:=\sqrt{(y_1\pm y_3)^2+y_2^2}.
\end{equation}
Each block has eigenvalues \(\mu\pm r_\pm\), hence its operator norm equals \(\max\{|\mu+r_\pm|,|\mu-r_\pm|\}=|\mu|+r_\pm\).
Therefore the full constraint is equivalent to
\begin{equation}
|\mu|+\max\{r_+,r_-\}\ \le\ \frac12.
\end{equation}
For any fixed \((y_1,y_2,y_3)\), the left-hand side is minimized at \(\mu=0\), so we may set \(\mu^\star=0\) and obtain the
equivalent constraints
\begin{equation}
r_+\le \frac12,
\qquad
r_-\le \frac12.
\label{eq:two_disks}
\end{equation}

\emph{Step 2: change of variables.}
Let \(u:=y_1+y_3\) and \(v:=y_1-y_3\). Then \eqref{eq:two_disks} becomes two disk constraints
\begin{equation}
u^2+y_2^2\le \frac14,
\qquad
v^2+y_2^2\le \frac14,
\label{eq:uv_disks}
\end{equation}
and the objective becomes
\begin{equation}
\bs\alpha\cdot\bs y
= \alpha_1\frac{u+v}{2}+\alpha_3\frac{u-v}{2}+\alpha_2 y_2
= p\,u + q\,v + \alpha_2 y_2,
\end{equation}
where \(p:=(\alpha_1+\alpha_3)/2\) and \(q:=(\alpha_1-\alpha_3)/2\).

\emph{Step 3: maximize over \(u,v\) for fixed \(y_2\).}
Fix \(y_2\) with \(|y_2|\le 1/2\) and set \(R:=\sqrt{\tfrac14-y_2^2}\).
Then \eqref{eq:uv_disks} implies \(|u|\le R\) and \(|v|\le R\), and the maximum of \(p\,u+q\,v\) is achieved at
\(u=R\,\mathrm{sgn}(p)\), \(v=R\,\mathrm{sgn}(q)\), yielding
\begin{equation}
\max_{u,v} \ (p\,u+q\,v+\alpha_2 y_2)
=
\alpha_2 y_2 + R\,(|p|+|q|).
\end{equation}
Thus
\begin{equation}
\gamma(\hat{\bs g}\mid\bs\alpha)
=
\max_{|y_2|\le 1/2}\Big(\alpha_2 y_2 + (|p|+|q|)\sqrt{\tfrac14-y_2^2}\Big).
\label{eq:1d_max}
\end{equation}

\emph{Step 4: solve the remaining one-dimensional maximization.}
Let \(y_2=\tfrac12\sin\theta\), so \(\sqrt{\tfrac14-y_2^2}=\tfrac12\cos\theta\).
Then the right-hand side of \eqref{eq:1d_max} becomes
\begin{equation}
\frac12\Big(\alpha_2\sin\theta + (|p|+|q|)\cos\theta\Big),
\end{equation}
whose maximum over \(\theta\) equals \(\tfrac12\sqrt{\alpha_2^2+(|p|+|q|)^2}\).
Finally, using the identity
\begin{equation}
|p|+|q|
=
\frac{|\alpha_1+\alpha_3|+|\alpha_1-\alpha_3|}{2}
=
\max\{|\alpha_1|,|\alpha_3|\},
\end{equation}
we obtain \eqref{eq:gamma_closed_2q}.
\end{proof}

\end{document}